\documentclass[11pt,letterpaper]{amsart}
\usepackage[foot]{amsaddr}
\usepackage[in]{fullpage}
\usepackage{amssymb}

\usepackage[T1]{fontenc}
\usepackage[mono=false]{libertine}
\usepackage{eulervm}
\usepackage{microtype}
\usepackage{hyperref}
\hypersetup{colorlinks=true, citecolor=blue, linkcolor = red, urlcolor=blue}
\usepackage{booktabs}
\newcommand{\ra}[1]{\renewcommand{\arraystretch}{#1}}
\renewcommand{\thefootnote}{\fnsymbol{footnote}}

\newtheorem{theorem}{Theorem}%[section]
\newtheorem*{theorem*}{Theorem}

\newtheorem{lemma}[theorem]{Lemma}
\newtheorem*{lemma*}{Lemma}
\newtheorem{proposition}[theorem]{Proposition}

\theoremstyle{definition}
\newtheorem{definition}{Definition}%[section]
\newtheorem*{definition*}{Definition}
\theoremstyle{remark}
\newtheorem{remark}{Remark}%[section]

\numberwithin{equation}{section}

\newcommand{\secref}[1]{Section~\ref{#1}}
\newcommand{\thmref}[1]{Theorem~\ref{#1}}
\newcommand{\lemref}[1]{Lemma~\ref{#1}}
\newcommand{\propref}[1]{Proposition~\ref{#1}}
\newcommand{\defnref}[1]{Definition~\ref{#1}}

\newcommand{\tabref}[1]{Table~\ref{#1}}

\newcommand{\prob}[1]{\ensuremath{\text{{Pr}$\left[#1\right]$}}}
\newcommand{\expct}[1]{\ensuremath{\text{{E}$\left[#1\right]$}}}
\newcommand{\var}[1]{\ensuremath{\text{{Var}$\left[#1\right]$}}}
\newcommand{\cov}[1]{\ensuremath{\text{{Cov}$\left[#1\right]$}}}

\begin{document}

\title{Understanding the hardness of approximate query processing with joins}

\author{Tianyu Liu$^1$}
\email{tl@cs.wisc.edu}
\address{$^1$University of Wisconsin--Madison, Madison, WI, USA}

\author{Chi Wang$^2$}
\email{wang.chi@microsoft.com}
\address{$^2$Microsoft Research, Redmond, WA, USA}

\begin{abstract}
We study the hardness of Approximate Query Processing (AQP) of various types of queries involving joins over multiple tables of possibly different sizes.
In the case where the query result is a single value (e.g., \texttt{COUNT}, \texttt{SUM}, and \texttt{COUNT(DISTINCT)}), 
we prove worst-case information-theoretic lower bounds for AQP problems that are given parameters $\epsilon$ and $\delta$, and return estimated results within a factor of $1 + \epsilon$ of the true results with error probability at most $\delta$.
In particular, the lower bounds for cardinality estimation over joins under various settings are contained in our results.
Informally, our results show that for various database queries with joins, unless restricted to the set of queries whose results are always guaranteed to be above a very large threshold, %$B$,
the amount of information an AQP algorithm needs for returning an accurate approximation is at least linear in the number of rows in the largest table.
Similar lower bounds even hold for some special cases where additional information such as top-$K$ heavy hitters and all frequency vectors are available.
In the case of \texttt{GROUP-BY} where the query result is not a single number, we study the lower bound for the amount of information used by any approximation algorithm that does not report any non-existing group and does not miss groups of large total size.
Our work extends the work of Alon, Gibbons, Matias, and Szegedy~\cite{Alon:1999:TJS:303976.303978}.

We compare our lower bounds with the amount of information required by Bernoulli sampling to give an accurate approximation.
For \texttt{COUNT} queries with joins over multiple tables of the same size, the upper bound matches the lower bound, unless the problem setting is restricted to the set of queries whose results are always guaranteed to be above a very large threshold. %$B$.
\end{abstract}

\maketitle

%\clearpage
%\setcounter{page}{1}

\let\thefootnote\relax\footnotetext{Research was done when Tianyu Liu was visiting Microsoft Research.}

\section{Introduction}
The database community has been working for more than 20 years on different techniques for Approximate Query Processing (AQP) that aim to deliver an estimate of the query result efficiently~\cite{10.1145/971697.602294, Olken93randomsampling}. In the past, different techniques for AQP have been proposed including approaches that leverage pre-computed samples or synopses~\cite{10.5555/1315451.1315455} as well as techniques that sample from the underlying data at query runtime (e.g. \cite{10.1145/93605.93611, 10.1006/jcss.1996.0041, 10.1145/335168.335230}). However, all the existing AQP approaches suffer from various limitations that restrict the applicability to support the ad-hoc exploration of a new data set.
In fact, even after decades of research, AQP remains largely confined to academic research and is not a well-established paradigm in today's products and services~\cite{Chaudhuri2017AQP}.

Among all the reasons that AQP is not widely adopted, an important one is that most AQP systems are not able to offer \emph{a priori} (i.e., before the query is executed by the AQP system) accuracy guarantee for arbitrary database queries. In other words, although certain AQP systems can provide accurate estimate efficiently for some queries, they could possibly return very bad estimate or take much longer time for other queries. As pointed out by \cite{Chaudhuri2017AQP}, "it seems impossible to
have an AQP system that supports the richness of SQL with significant
saving of work while providing an accuracy guarantee that
is acceptable to a broad set of application workloads."

Together with the unsatisfactory development of AQP systems that can provide guaranteed error bounds is the lack of understanding toward whether such systems are even possible.
Except for very few research papers, e.g. \cite{Alon:1999:TJS:303976.303978}, not much is known about the lower bound, i.e., how much information we need to know ahead to give an estimate for common database queries with given accuracy requirement.

In this paper, we try to answer this question mainly for various common database queries involving joins which return a single value (e.g., \texttt{COUNT}).
This is done by proving worst-case information-theoretic lower bounds for AQP problems that are given parameters $\epsilon$ and $\delta$, and return estimated results within a factor of $1 + \epsilon$ of the true results with error probability at most $\delta$ (see \defnref{defn:approx}). 
Our work extends the work of Alon, Gibbons, Matias, and Szegedy~\cite{Alon:1999:TJS:303976.303978}.
Compared with \cite{Alon:1999:TJS:303976.303978} which only considered the lower bound of join size estimation (essentially only \texttt{COUNT} queries) over two tables of the same size with a constant relative error and a constant error probability, we show lower bounds for a much broader range of query types involving joins, including \texttt{COUNT}, \texttt{SUM}, \texttt{COUNT(DISTINCT)}, \texttt{GROUP-BY}, etc. over multiple tables of possibly different sizes.
In the case where the return value is a single value (e.g., \texttt{COUNT}, \texttt{SUM}, and \texttt{COUNT(DISTINCT)}), our results allow relative error parameter $\epsilon$ and error probability $\delta$.
In particular, the lower bounds for cardinality estimation over joins under various settings are contained in our results.
In the case of \texttt{GROUP-BY} where the return value is not a single number, we study the lower bound for the amount of information used by any approximation algorithm that does not report any non-existing group and does not miss groups of large total size.
For most type of queries, we prove lower bounds for queries without selection.
Nevertheless, the lower bounds still hold when \texttt{WHERE} clauses are allowed, as the existence of selection does not make the problems easier.
For PK-FK joins, we assume selection is present because when no selection is applied on the query, the result of query function such as \texttt{COUNT} only depends on the fact table, and hence lower bounds for a single relation (without join) holds~\cite{Ding:2016:SSA:2882903.2915249}.

As the theorems we prove in this paper are information-theoretic lower bounds, we use the abstract notions that an algorithmic scheme $\Phi$ extracts information from each of the database relations (oblivious to other database relations) as bit strings, and another function $D$ computes the query results from the the bit strings.
Our theorems apply to (but are not restricted to) common AQP implementations of $\Phi$ such as computing histograms or drawing samples of database relations without information from other relations.
In addition to algorithmic schemes that can only look at each relation independently,
there are research works in the AQP literature which assume additional information such as top-K heavy hitters and frequency vectors of each table are available (e.g. \cite{Estan:2006:ESJ:1129754.1129881, Vengerov:2015:JSE:2824032.2824051, Chen:2017:TSJ:3035918.3035921}).
In this paper, we also consider such situations.

Since in most types of queries we study (except for \texttt{GROUP-BY}), our notion of approximation is to find a value within a small factor $1+\epsilon$ of the true value, it is extremely hard if the true value itself is very small. For example, in the case of \texttt{COUNT}, if the true answer is just one, i.e. there is exact one database record satisfying the query, then any approximation algorithm with small $\epsilon$ has to find this record.
This makes the problem as hard as finding \texttt{MAX} or \texttt{MIN} which requires looking at the full table.
Therefore, in our setting we assume there is a given value $B$ such that the answer for any query is above $B$.
Intuitively, the problem is more difficult when $B$ is smaller.

Informally, our results show that for various database queries with joins (except for PK-FK joins), unless restricted to the set of queries whose results are always guaranteed to be above a very large $B$, the amount of information an AQP algorithm needs for returning an $(\epsilon, \delta)$-approximation (\defnref{defn:approx}) is at least linear in the number of rows in the largest table.
Similar lower bounds even hold for some special cases where additional information such as top-K heavy hitters and all frequency vectors are available.
Our lower bounds hold even when each database relation has only two columns.
As AQP systems only need the information in the joining columns as well as the columns where query functions (e.g. \texttt{SUM} or \texttt{COUNT(DISTINCT)}) apply, other columns can be discarded.
In this paper, we use interchangeably the number of rows in a database and the size of the database relation.
Our results, to some extent, indicate that general AQP systems with predetermined accuracy are difficult to build as the amount of space and time they require would be comparable to that required by systems returning exact answers. 

However, for PK-FK joins we are only able to prove a lower bound that is linear in the size of the largest dimension table. Since most of the time it is the fact table that has enormous size, our theorem does not preclude a relatively efficient AQP system for PK-FK joins.

Our proof is built on the argument in \cite{Alon:1999:TJS:303976.303978} which combines Yao's principle~\cite{4567946} and the probabilistic method. In order to support the quantitative statements with extra parameters and various common query types, we use a more accurate estimate of binomial coefficients and carry out more careful mathematical arguments than those in \cite{Alon:1999:TJS:303976.303978}.

To have a better idea of how tight our lower bounds are, we also compare them with the amount of information required by Bernoulli sampling to give an $(\epsilon, \delta)$-approximation with constant $\epsilon$ and $\delta$.
For \texttt{COUNT} queries with joins over multiple tables of the same size, the upper bound matches the lower bound, unless the problem setting is restricted to the set of queries whose results are always guaranteed to be above a very large $B$.

The above comparison not only implies that our hardness result is tight, but also suggests that any effort aiming at creating an AQP system with guaranteed accuracy bound and the support for general query types with joins and unrestricted data distribution cannot perform better (by more than a constant factor) than Bernoulli sampling in terms of time and space complexity.

In fact, our results do not rule out the possibility to build efficient AQP systems with guaranteed accuracy, but one might need to sacrifice in one of the following aspects:
\begin{itemize}
\item
Restrict the set of queries or data such that any query will return a large number.
In \cite{Ding:2016:SSA:2882903.2915249}, the authors built a two-layered system in which the first layer efficiently deals with ``frequent'' data and infrequent data is left for the second layer possibly taking longer time.

\item
Restrict the type of queries or have more knowledge about the data.
In \cite{Chen:2017:TSJ:3035918.3035921}, the authors designed new sampling-based approximation algorithms for PK-FK joins, chain joins, and star joins, and showed that they outperform Bernoulli sampling in these settings, respectively.
They also studied the restricted setting where frequencies are known.
Related work in this regard is further discussed in \secref{sec:related_works}.

\end{itemize}

The layout of this paper is as follows.
In \secref{sec:related_works}, we discuss research work related to our paper.
The lower bounds for most common query types, including \texttt{COUNT}, \texttt{SUM}, \texttt{COUNT(DISTINCT)}, and \texttt{GROUP-BY}, is presented in \secref{sec:hardness_common}.
The hardness of AQP with additional restrictions on queries (PK-FK joins) or additional information on data (top-K heavy hitters and frequency vectors) is discussed in \secref{sec:hardness_additional}.
In \secref{sec:comparison}, we compare our lower bounds on \texttt{COUNT} queries with the amount of resource required by Bernoulli sampling.

\bigskip

\section{Related Work}\label{sec:related_works}
Apart from \cite{Alon:1999:TJS:303976.303978}, there is not much theoretical work on the hardness of AQP or cardinality estimation over joins.
Recently, Huang et al.~\cite{10.14778/3372716.3372726} proved an information-theoretic lower bound on the lowest variance achievable by any sampling strategy in the two-table case for query functions including \texttt{COUNT}, \texttt{SUM}, and \texttt{AVG}.
Their work is based on results in communication complexity theory on set intersection.
Assume that both Alice and Bob each hold a set of size $k$, say $A$ and $B$, respectively, and they aim to estimate the size of $t = |A \cap B|$.
In the one-way model, Alice computes a summary $\beta(A)$ and sends it to Bob who will estimate $|A \cap B|$ using $B$ and $\beta(A)$.
Pagh et. al~\cite{10.1145/2594538.2594554} showed that any one-way communication protocol that estimates $t$ within a factor between $1 - \delta$ and $1 + \delta$ with probability at least $2/3$ must send at least $\Omega(k/(t\delta^2))$ bits.
This served as a reduction source in \cite{10.14778/3372716.3372726} for \texttt{COUNT}, \texttt{SUM}, and \texttt{AVG} for which similar bounds were given.
They also designed a hybrid sampling scheme that combines stratified, universe, and Bernoulli sampling, and showed that in the two-table case with the frequency vectors, this scheme can achieve the theoretical lower bound within a constant factor.

Most of the lower bounds in this paper are based on the assumption that algorithmic schemes extract information from each database relation without access to other relations, except in \thmref{thm:heavy_hitter} and \thmref{thm:frequency} where we allow heavy hitters and frequency vectors, respectively, on the joining columns.
This assumption is applicable to most mainstream works on AQP or cardinality estimation over joins.
One category of approaches such as \cite{Alon:1999:TJS:303976.303978, 10.5555/982792.982884, CORMODE200558, 10.5555/1083592.1083598, 10.1145/1386118.1386121} is \emph{sketching-based} where a sketch for each table on the join attribute is built while all the other attributes are ignored.
The other major category consists of \emph{sampling-based} algorithms, including 
ripple join~\cite{10.1145/304182.304208, 10.1145/564691.564721}, block-level sampling~\cite{10.1145/1007568.1007602}, bi-level Bernoulli sampling~\cite{10.1145/1007568.1007601}, correlated sampling~\cite{Vengerov:2015:JSE:2824032.2824051}, etc.
A set of research work only supports foreign key joins with the same assumption, e.g., join synopses~\cite{10.1145/304182.304207}, AQUA~\cite{10.5555/645925.671347}, Icicles~\cite{10.5555/645926.672017}, STRAT~\cite{10.1145/1242524.1242526}, and BlinkDB~\cite{10.1145/2465351.2465355}.
Besides, end-biased sampling~\cite{Estan:2006:ESJ:1129754.1129881} samples each tuple with a probability proportional to the frequency of its join key.

We also note that there exists research work on AQP where additional information is assumed to be available.
For instance, there are many sampling-based algorithms that rely on indexes, including adaptive sampling~\cite{LIPTON199518}, bifocal sampling~\cite{10.1145/233269.233340}, and wander join~\cite{10.1145/2882903.2915235}. The use of indexes allows the sampling to be much focused, retrieving only tuples that are relevant to the query.
Recently, machine learning was introduced to address the problem of cardinality estimation~\cite{DBLP:conf/cidr/KipfKRLBK19,dutt2019selectivity}.
In addition, \cite{Chaudhuri:1999:RSO:304182.304206, Zhao:2018:RSO:3183713.3183739} studied the problem of getting a simple random sample of the full join results.
Hardness of AQP under these circumstances is to be addressed in the future.

\bigskip

\section{Hardness of AQP in general}\label{sec:hardness_common}
In order to prove the theorems in this section and \secref{sec:hardness_additional}, we first prove the following key technical lemma. 

\begin{definition}
Let $k, \alpha$ be positive integers and $\beta$ be a positive real number.
A \emph{$(k, \alpha, \beta)$-set} $\mathcal{S}$ is a family of subsets of $\{1, 2, \dots, \alpha k\}$ such that:
\begin{enumerate}
\item $|S| = k, \forall S \in \mathcal{S}$;
\item $|\mathcal{S}| = 2^{\beta k}$;
\item $|S_i \bigcap S_j| \le \frac{k}{2}, \forall S_i, S_j \in \mathcal{S}$.
\end{enumerate}
\end{definition}
We remark that in the above definition, the second item should be $|\mathcal{S}| = \left\lfloor 2^{\beta k} \right\rfloor$. Similarly, in the statement of our theorems, the information-theoretic lower bounds in the number of bits should be rounded down to the nearest integers, which we omit for simplicity.

The following lemma uses the \emph{probabilistic method} and an accurate estimate of binomial coefficients to show when a $(k, \alpha, \beta)$-set exists.
\begin{lemma}\label{lem:set}
Let $\alpha \ge 2$ and $\beta$ be constant numbers. There exists a $k_0(\alpha, \beta)$ such that a $(k, \alpha, \beta)$-set exists for any $k \ge k_0$
if $\beta < \frac{1}{2}\left(\alpha H\left(\frac{1}{\alpha}\right) - (\alpha-1) H\left(\frac{1}{2(\alpha-1)}\right) - 1\right)$.
\end{lemma}
\begin{proof}
In order to show the existence of such a family of subsets, we instead show that for a randomly chosen set $\mathcal{S}'$ of subsets of $\{1, 2, \dots, \alpha k\}$ satsifying (1) and (2), the probability that (3) is violated is less than 1. As a consequence, the probability of the existence of a set satisfying (1), (2), and (3) simultaneously is larger than 0.

For any two randomly chosen subsets $S_i, S_j \in \mathcal{S}'$,
\begin{equation}\label{eqn:binomial}
\prob{|S_i \bigcap S_j| = l} = \frac{{k \choose l}{\alpha k - k \choose k - l}}{{\alpha k \choose k}}.
\end{equation}
It is easy to verify that for any $\frac{k}{2} < l \le k$,
\[
\prob{|S_i \bigcap S_j| = l} < \prob{|S_i \bigcap S_j| = \frac{k}{2}},
\]
thus for $S_i$ and $S_j$ the probability that (3) is violated is
\begin{equation}\label{eqn:single_pair}
\prob{|S_i \bigcap S_j| > \frac{k}{2}} = \sum\limits_{l=k/2 +1}^{k} \prob{|S_i \bigcap S_j| = l} < \frac{k}{2} \cdot \prob{|S_i \bigcap S_j| = \frac{k}{2}}.
\end{equation}

According to (\ref{eqn:binomial}), $\prob{|S_i \bigcap S_j| = \frac{k}{2}} = \frac{{k \choose k/2}{\alpha k - k \choose k/2}}{{\alpha k \choose k}}.$
To get an accurate estimation of it, we use the following asymptotics from \cite{binomial_coefficients} when $y = \Omega(x)$, i.e. when $y$ is linear in $x$:
\[
\log {x \choose y} = (1 + o(1)) \ H\left(\frac{y}{x}\right) x,
\]
where $H(p) = -p\log p - (1-p) \log (1-p)$ is the \emph{binary entropy function}.
Therefore, combining with the fact that ${k \choose k/2} \le 2^k$, we have the following when $\alpha$ is a constant:
\begin{equation}\label{eqn:special}
\begin{split}
\prob{|S_i \bigcap S_j| = \frac{k}{2}} 
& \le 2^{o(k)} \cdot \frac{2^k \cdot 2^{H\left(\frac{1}{2(\alpha-1)}\right) (\alpha-1)k}}{2^{H\left(\frac{1}{\alpha}\right) \alpha k}} \\
& = 2^{o(k)} \cdot 2^{\left(1 + (\alpha-1) H\left(\frac{1}{2(\alpha-1)}\right) - \alpha H\left(\frac{1}{\alpha}\right) \right)k}.
\end{split}
\end{equation}

By the \emph{union bound}, the probability that there exists some pair of subsets violating (3) is at most
\begin{equation}\label{eqn:union}
\sum_{S_i, S_j \in {\mathcal{S} \choose 2}} \prob{|S_i \bigcap S_j| > \frac{k}{2}} \le 2^{2\beta k} \cdot \prob{|S_i \bigcap S_j| > \frac{k}{2}}
\end{equation}

Combining (\ref{eqn:single_pair}), (\ref{eqn:special}), and (\ref{eqn:union}), 
this probability that the ``bad'' event happens is less than 1 for $k$ large enough if
\[
1 + (\alpha-1) H\left(\frac{1}{2(\alpha-1)}\right) - \alpha H\left(\frac{1}{\alpha}\right) + 2\beta < 0.
\]
\end{proof}

In order to give a lower bound on the amount of resources that any randomized algorithm has to cost in the worst input case, according to \emph{Yao's principle} for Monte Carlo algorithms~\cite{4567946}, we can instead study the amount of resources that any deterministic algorithm has to cost in the worst case of input distribution, thanks to the following theorem.
The cost could be measured in terms of time complexity, space complexity, or any other quantity that describes an algorithm.
In this paper, the cost will be the amount of information measured in terms of the length of bit strings.

\begin{theorem}[Yao~\cite{4567946}]\label{thm:yao}
Denote by $E(R, x)$ to be the expected cost of a randomized algorithm $R$ on input $x$. Let $\mathcal{R}_\delta$ be the set of randomized algorithms that make mistakes with probability no more than $\delta$ on the worst input case.
Denote by $C(A, d)$ to be the average cost of a deterministic algorithm $A$ on a probability distribution of input $d$. Let $\mathcal{A}_\delta$ be the set of deterministic algorithms that make mistakes on no more than $\delta$ proportion of the input in distribution $d$.
Then
\[
\min_{R \in \mathcal{R}_\delta} \max_x E(R, x) \ge \frac{1}{2} \max_d \min_{A \in \mathcal{A}_{2\delta}} C(A, d).
\]
\end{theorem}

Next we define the error metric to be used on the approximate result of aggregate functions.

\begin{definition}\label{defn:approx}
Suppose $f: \Sigma^* \rightarrow \mathbb{R}$ is a function mapping problem instances to real numbers. An \emph{$(\epsilon, \delta)$-approximation} for $f$ is a randomized algorithm that takes as input an instance $x$, and outputs a number $Y$ (a random variable) such that 
\[
\operatorname{Pr}\left[\frac{1}{1+\epsilon} f(x) \le Y \le (1+ \epsilon) f(x)\right] \ge 1- \delta.
\]
\end{definition}

According to \thmref{thm:yao}, in order to show a lower bound on the amount of resources that any randomized algorithm has to use to give an $(\epsilon, \delta)$-approximation, we only need to find a ``bad'' input distribution for which no deterministic algorithm could return a ``good'' result (within a ratio of $1+\epsilon$) for more than $1 - 2\delta$ proportion of the input.
In this paper, we will assume the error probability $\delta < 6.25\%$ for technical reasons.

\begin{definition}
Given $\delta < 6.25\%$, define $C(\delta) = \frac{H(8\delta) - (1-8\delta)H\left(\frac{4\delta}{1-8\delta}\right)}{16 \delta} - \frac{1}{2}$, where $H(p) = -p \log p - (1-p) \log (1-p)$ is the entropy function.
\end{definition}

\begin{remark}
When $1 - \delta = 95\%$, $C(\delta) > \frac{1}{50}$; when $1 - \delta = 99\%$, $C(\delta) > \frac{1}{2}$; when $1 - \delta = 99.9\%$, $C(\delta) > 3.6$.
\end{remark}

\begin{theorem}[2-table \texttt{COUNT}]\label{thm:2-table}
Let $\Phi$ be any scheme which assigns bit strings to database relations, so that there is a function $D$ such that given two relations $R_1$ of size $n_1$ and $R_2$ of size $n_2$, $D(\Phi(R_1), \Phi(R_2))$ gives an $(\epsilon, \delta)$-approximation for the \texttt{COUNT} result of $R_1 \Join R_2$ (join size), when an a priori lower bound $B < \frac{n_1n_2}{(1+\epsilon)^2}$ is given on the join size. Then the length of the bit string that $\Phi$ assigns must be at least
\[
\frac{C(\delta)}{2} \cdot \min\left(\max(m_1, m_2), \frac{m_1 m_2}{(\epsilon^2 + 2\epsilon) B}\right),
\] where $m_1 = n_1\left(1 - \left(\frac{B}{n_1 n_2}\right)^{1/2}\right)$, and $m_2 = n_2\left(1 - \left(\frac{B}{n_1 n_2}\right)^{1/2}\right)$.
\end{theorem}

Before proving the theorem, let us briefly explain the role played by the lower bound $B$ on the join size.

Since the join size is at most $n_1n_2$, the requirement $B < \frac{n_1n_2}{(1+\epsilon)^2}$ is necessary, because otherwise any value in $\left[\frac{n_1n_2}{1+\epsilon}, (1+\epsilon)B\right]$ is a good estimate, destroying the purpose of any approximation scheme.
This should explain similar restrictions on $B$ in most of the theorems in this paper.

In $m_1$ and $m_2$, $\frac{B}{n_1 n_2}$ can be viewed as a lower bound on the \emph{selectivity} of the queries under consideration.
It says that when $B$ decreases, $m_1$ and $m_2$ increase, i.e., the more selective the queries are, the more information an algorithm needs to give precise estimations of the query results, which aligns with our intuition.

If we assume $m_1 \le m_2$, the threshold between $\max(m_1, m_2)$ and $\frac{m_1 m_2}{(\epsilon^2 + 2\epsilon) B}$ in the theorem statement is $B = \frac{m_1}{\epsilon^2 + 2\epsilon}$. Above this value, the higher $B$ is, the less our lower bound is. This is the range in which AQP systems can potentially save time and space resources compared with systems returning exact values. However, below this value, the lower bound is roughly linear in $n_2$, which says a scheme need information no less than the larger table size. Within this range, AQP systems might not be able to dominate other systems in terms of performance.

\begin{proof}[Proof of \thmref{thm:2-table}]
Without loss of generality, assume $n_1 \le n_2$ (and equivalently $m_1 \le m_2$).
We also assume each of $R_1$ and $R_2$ has at least two columns, and the joining operation is on the column $R_1.C$ (i.e. $R_2.C$).
Define $k$ such that
\begin{itemize}
\item if $m_1 < (\epsilon^2 + 2\epsilon) B$, let $k = \frac{m_1 m_2}{(\epsilon^2 + 2\epsilon) B}$;
\item otherwise, let $k = m_2$.
\end{itemize}
In other words, $k = \min\left(m_2, \frac{m_1 m_2}{(\epsilon^2 + 2\epsilon) B}\right)$.
Let $t = \frac{k}{8\delta}$.
Fix a set $T = \{1, 2, \dots, t\}$ of $t$ possible values for column $C$, denoted by \emph{types}. Also assume there is an extra type $0$ of value.
Define $\alpha = \frac{1}{8\delta}$ and fix $\beta < \frac{1}{2}\left(\alpha H\left(\frac{1}{\alpha}\right) - (\alpha-1) H\left(\frac{1}{2(\alpha-1)}\right) - 1\right)$.
According to \lemref{lem:set}, there exists a $(k, \alpha, \beta)$-set $\mathcal{S}$ as a family of subsets of $T$.

Now we are ready to define the ``bad'' distribution on inputs for $R_1$ and $R_2$.
Let $d_1$ be a probabilistic distribution on relations of size $n_1$ such that
\begin{itemize}
\item the first $m_1$ entries on column $C$ takes values from $T$ uniformly at random;
\item the remaining $\sqrt{\frac{n_1}{n_2} \cdot B}$ entries on column $C$ take the value of type $0$.
\end{itemize}
Let $d_2$ be the uniform distribution on relations of size $n_2$ such that
\begin{itemize}
\item the first $m_2$ entries on column $C$ are partitioned into $k$ parts of equal size $\frac{m_2}{k}$ and equal value and the $k$ different partitions take $k$ corresponding different types of values from a uniformly chosen subset $S \in \mathcal{S}$;
\item the remaining $\sqrt{\frac{n_2}{n_1} \cdot B}$ entries on column $C$ take the value of type $0$.
\end{itemize}

When $R_1$ and $R_2$ are independently from $d_1$ and $d_2$, respectively, the join size is either exactly $B$ (if the first $m_1$ entries in $R_1$ fail to join any any of the first $m_2$ entries in $R_2$), or at least $(1 + \epsilon)^2 B$.

Consider partitioning the relations into classes according to the bit string assigned by $\Phi$. For each relation in $d_1$, the function $D$ gives the same estimate for all relations in $d_2$ in the same class. However, for each class, there can be at most one relation in $d_2$ for which the estimate is within $\left(\frac{1}{1+\epsilon}, 1+\epsilon\right)$ ratio of the true value for more than $1 - 2\delta$ of the relations in $d_1$.

To see this, consider $S_i, S_j \in \mathcal{S}$ such that the corresponding relations in $d_2$ map to the same class, and let $T' = S_1 \oplus S_2$. For each relation $R_1$ from $d_1$ whose column $C$ has value in $T'$, the join size is $B$ for one of $S_i, S_j$ and at least $(1 + \epsilon)^2 B$ for the other.
Thus any estimate will be at least $(1 + \epsilon)$ away for at least one of them. 
By properties of $\mathcal{S}$, $|S_1 \setminus S_2| \ge \frac{k}{2}$ and $|S_2 \setminus S_1| \ge \frac{k}{2}$, and hence for one of them, the estimate will have at least a $(1 + \epsilon)$ factor of error for more than $\frac{k}{2} / \frac{k}{8\delta} = 4\delta$ of the relations in $d_1$.

When the length of bit strings is less than $\beta k$, i.e. at most $\beta k - 1$, then the number of distinct classes it can express is at most $\frac{1}{2}$ of $|\mathcal{S}|$.
Thus the proportion of input that the estimate will have an error of at least a $1 + \epsilon$ factor is at least $4\delta \cdot \frac{1}{2} = 2\delta$.
Therefore, the length of the bit strings should be at least $\beta k = C(\delta) \cdot \min\left(m_2, \frac{m_1 m_2}{(\epsilon^2 + 2\epsilon) B}\right)$.
Applying \thmref{thm:yao} gives the result.
\end{proof}

\begin{theorem}[multi-table \texttt{COUNT}]\label{thm:multi-table}
Let $\Phi$ be any scheme which assigns bit strings to database relations, so that there is a function $D$ such that given $p$ relations $R_1$ of size $n_1$, $R_2$ of size $n_2$, \dots, and $R_p$ of size $n_p$, $D(\Phi(R_1), \Phi(R_2), \cdots, \Phi(R_p))$ gives an $(\epsilon, \delta)$-approximation for the \texttt{COUNT} result of $\Join_i R_i$ (join size), when an a priori lower bound $B < \frac{\prod_i n_i}{(1+\epsilon)^2}$ is given on the join size. Then the length of the bit string that $\Phi$ assigns must be at least
\[
\frac{C(\delta)}{2} \cdot \min\left( \max_i m_i, \frac{\prod_i m_i}{(\epsilon^2 + 2\epsilon) B}\right),
\] where $m_i = n_i\left(1 - \left(\frac{B}{\prod_i n_i}\right)^{1/p}\right)$.
\end{theorem}
\begin{proof}
Without loss of generality, assume $n_p \ge n_i$ for any $1 \le i \le p-1$.
Again we assume each of $R_i$ has at least two columns.
%Given any \emph{join graph}, denote the neighbors of $R_p$ in the join graph by $N_{R_p}$.
%
Define $k$ such that
\begin{itemize}
\item if $\prod_{i=1}^{p-1} m_i < (\epsilon^2 + 2\epsilon) B$, let $k = \frac{\prod_i m_i}{(\epsilon^2 + 2\epsilon) B}$;
\item otherwise, let $k = m_p$.
\end{itemize}
In other words, $k = \min\left(m_p, \frac{\prod_i m_i}{(\epsilon^2 + 2\epsilon) B}\right)$.
Also define $t, \alpha, \beta$, and the set $T, \mathcal{S}$ similarly as in \thmref{thm:2-table}.

Consider the following input distribution.
Let $d_1$ be a joint distribution on the joining columns of tables $R_1, R_2, \dots, R_{p-1}$ such that
\begin{itemize}
\item the first $m_i$ entries in $R_i$ take the same value uniformly randomly chosen from $T$;
\item the other entries in each table take the value of type $0$.
\end{itemize}
Let $d_2$ be a joint distribution on the joining columns of table $R_p$ such that
\begin{itemize}
\item the first $m_p$ entries are partitioned into $k$ parts of equal size $\frac{m_p}{k}$ and equal value, and the $k$ different partitions take $k$ corresponding different types of values from a uniformly chosen subset $S \in \mathcal{S}$ on each joining column;
\item the other entries on column $C$ take the value of type $0$.
\end{itemize}
The rest of the proof follows that in \thmref{thm:2-table}.
\end{proof}

\begin{theorem}[multi-table \texttt{SUM}]
Let $\Phi$ be any scheme which assigns bit strings to database relations, so that there is a function $D$ such that given $p$ relations $R_1$ of size $n_1$, $R_2$ of size $n_2$, \dots, and $R_p$ of size $n_p$, $D(\Phi(R_1), \Phi(R_2), \cdots, \Phi(R_p))$ gives an $(\epsilon, \delta)$-approximation for the \texttt{SUM} result of $\Join_i R_i$, when an a priori lower bound $B < \frac{M \cdot \prod_i n_i}{(1+\epsilon)^2}$ is given on the \texttt{SUM} result and the maximum value $M$ is known on the \texttt{SUM} column. Then the length of the bit string that $\Phi$ assigns must be at least
\[
\frac{C(\delta)}{2} \cdot \min\left( \max_i m_i, \frac{M \cdot \prod_i m_i}{(\epsilon^2 + 2\epsilon) B}\right),
\] where $m_i = n_i\left(1 - \left(\frac{B}{M \cdot \prod_i n_i}\right)^{1/p}\right)$.
\end{theorem}
\begin{remark}
If \texttt{SUM} is on the results of a function $f$ (e.g., product) across multiple columns, the theorem still applies if $M$ denotes the maximum value in the range of $f$.
\end{remark}
\begin{proof}
First we assume the \texttt{SUM} column is not among the joining columns.
Following the settings in the proof of \thmref{thm:multi-table}, define $k$ such that
\begin{itemize}
\item if $\prod_{i=1}^{p-1} m_i < (\epsilon^2 + 2\epsilon) B / M$, let $k = \frac{M \cdot \prod_i m_i}{(\epsilon^2 + 2\epsilon) B}$;
\item otherwise, let $k = m_p$.
\end{itemize}
In other words, $k = \min\left(m_p, \frac{\prod_i m_i}{(\epsilon^2 + 2\epsilon) B}\right)$.
The rest of the proof is similar to that in \thmref{thm:multi-table}, except that every value in the \texttt{SUM} column is replaced with $M$.
\end{proof}

\begin{theorem}[multi-table \texttt{COUNT(DISTINCT)}]
Following the setting in \thmref{thm:multi-table}, suppose there are at least two tables joining together. Let $R$ be the table where \texttt{COUNT(DISTINCT)} applies.
Suppose the size of $R$, $n > (1+\epsilon)^2 B$.
Then the length of the bit string that $\Phi$ assigns must be at least
\[
\frac{C(\delta)}{2} \cdot (n - B) \min\left(1, \frac{1}{(\epsilon^2 + 2\epsilon) B}\right).
\]
\end{theorem}
\begin{remark}
In this theorem, we assume \texttt{COUNT(DISTINCT)} is applied on a single column or multiple columns on a single table. When \texttt{COUNT(DISTINCT)} can be applied to multiple columns across multiple tables, the result is similar to that in \thmref{thm:multi-table} (as we can construct the input such that every row in the join result is distinct).
\end{remark}
\begin{proof}
Construct relation $R$ such that: (1) there are at least two columns in $R$; (2) the \texttt{COUNT(DISTINCT)} columns do not have intersection with the joining columns; (3) every row in $R$ is distinct from other rows in terms of the \texttt{COUNT(DISTINCT)} columns.
Define $k$ such that
\begin{itemize}
\item if $1 < (\epsilon^2 + 2\epsilon) B$, let $k = \frac{n - B}{(\epsilon^2 + 2\epsilon) B}$;
\item otherwise, let $k =  n - B$.
\end{itemize}
In other words, $k = (n - B) \min\left(1, \frac{1}{(\epsilon^2 + 2\epsilon) B}\right)$.

The rest of the proof is similar to those in \thmref{thm:2-table} and \thmref{thm:multi-table} with slight modifications.
In terms of the input distributions, let $R'$ be one of the tables that join with $R$ and assume $R'$ has at least two rows. The distribution $d_1$ on the joining column of $R'$ has one entry taking a uniformly randomly chosen value from $T$ and another entry of value of type $0$. The distribution $d_2$ on the joining column of $R$ has $(n-B)$ entries partitioned into $k$ classes (with $k$ different values from a random set uniformly chosen from $\mathcal{S}$ as usual) and another $B$ entries of value of type $0$.
\end{proof}

\begin{theorem}[multi-table \texttt{GROUP-BY}]\label{thm:group-by}
Let $\lambda \ge 1$.
Let $\Phi$ be any scheme which assigns bit strings to database relations, so that there is a function $D$ such that given $p$ relations $R_1$ of size $n_1$, $R_2$ of size $n_2$, \dots, and $R_p$ of size $n_p$, $D(\Phi(R_1), \Phi(R_2), \cdots, \Phi(R_p))$ reports all the groups 
in the \texttt{GROUP-BY} result of $R_1 \Join R_2$ with one-sided error, i.e. it does not report any non-existing group, then in order not to miss a set of groups whose total size is no larger than $\lambda \cdot \frac{\prod_in_i}{\max_i n_i}$ with $1-\delta$ probability, the length of the bit string that $\Phi$ assigns must be at least
\[
\frac{C(\delta)}{2} \cdot \frac{\max_i n_i}{\lambda}.
\]
\end{theorem}
\begin{proof}
For simplicity of discussion, we prove the theorem for the two-table case (with $n_2 \ge n_1$) where a single group of size $\lambda n_1$ should not be missed. The extension for multiple tables can be achieved similar to the extension from \thmref{thm:2-table} to \thmref{thm:multi-table}; the extension for multiple groups can be achieved by adding a \texttt{GROUP-BY} column and subdividing a single group.

In the two-table case, let $k = \frac{n_2}{\lambda}$. Also define $t, \alpha, \beta$, and construct $T, \mathcal{S}$ similarly.
Let $d_1$ be the distribution on $R_1$ such that all entries on the joining column are of the same value uniformly chosen from $T$. Let $d_2$ be the distribution on $R_2$ such that entries on the joining column are partitioned into $k$ classes of the same size $\lambda$ and entries in each class have the same value, and the $k$ different values are from a uniformly chosen set $S \in \mathcal{S}$. 

When $R_1$ and $R_2$ are independently from $d_1$ and $d_2$, respectively, the join size is either $0$ or $\lambda n_1$. Suppose in the latter case, the results are always in the same single group. Then in the result, there is either no group existing, or a group of size $\lambda n_1$.

Consider partitioning the relations into classes according to the bit string assigned by $\Phi$. For each relation in $d_1$, the function $D$ gives the same estimate for all relations in $d_2$ in the same class. However, for each class, there can be at most one relation in $d_2$ for which the no group is missed for more than $1 - 2\delta$ of the relations in $d_1$.

To see this, consider $S_i, S_j \in \mathcal{S}$ such that the corresponding relations in $d_2$ map to the same class, and let $T' = S_1 \oplus S_2$. For each relation $R_1$ from $d_1$ whose column $C$ has value in $T'$, the join result is ``no group exists'' for one of $S_i, S_j$ and ``there exists a group of size $\lambda n_1$'' for the other.
Since we assume the approximation algorithms only have one-sided error, i.e., it does not report any non-existing group, any such algorithm need to report ``no group exists'' for relations from $T'$.
By properties of $\mathcal{S}$, $|S_1 \setminus S_2| \ge \frac{k}{2}$ and $|S_2 \setminus S_1| \ge \frac{k}{2}$, and hence for one of them, the estimate will miss a group of size $\lambda n_1$ for more than $\frac{k}{2} / \frac{k}{8\delta} = 4\delta$ of the relations in $d_1$.

When the length of bit strings is less than $\beta k$, i.e at most $\beta k - 1$, then the number of distinct classes it can express is at most $\frac{1}{2}$ of $|\mathcal{S}|$.
Thus the proportion of input that the estimate will miss a group of size $\lambda n_1$ is at least $4\delta \cdot \frac{1}{2} = 2\delta$.
Therefore, the length of the bit strings should be at least $\beta k = C(\delta) \cdot \frac{n_2}{\lambda}$.
Applying \thmref{thm:yao} gives the result.
\end{proof}

\bigskip

\section{Hardness of AQP with additional restrictions or information}\label{sec:hardness_additional}
The lower bounds in \secref{sec:hardness_common} on various type of query functions still hold when \texttt{WHERE} clauses are allowed, as the existence of selection does not make the problems easier.
In \thmref{thm:pk-fk} and \thmref{thm:pk-fk_group-by}, we study PK-FK join.
We note that when no selection is applied on the query, the result of query function such as \texttt{COUNT} only depends on the fact table, and hence lower bounds for a single relation (without join) holds.
Therefore, our results assume that selection is present.

\begin{definition}
Given $\delta < 6.25\%$, define $C'(\delta) = 8\delta \cdot C(\delta)$.
\end{definition}

\begin{remark}
When $1 - \delta = 95\%$, $C'(\delta) > \frac{1}{125}$; when $1 - \delta = 99\%$, $C'(\delta) > \frac{1}{25}$.
\end{remark}

\begin{theorem}[PK-FK join with selection \texttt{COUNT}]\label{thm:pk-fk}
Let $\Phi$ be any scheme which assigns bit strings to database relations, so that there is a function $D$ such that given a dimension table of size $n_D$ and a fact table of size $n_F$, $D(\Phi(R_1), \Phi(R_2))$ gives an $(\epsilon, \delta)$-approximation for the \texttt{COUNT} result of $R_1 \Join R_2$ (join size), when an a priori lower bound $B < \frac{n_F}{(1+\epsilon)^2}$ is given on the join size. Then the length of the bit string that $\Phi$ assigns must be at least
\[
\frac{C'(\delta)}{2} \cdot (n_D - 1).
\]
\end{theorem}
\begin{remark}
For the star schema with one fact table of size $n_F$ and $p$ dimension tables of sizes $n_{D_1}, n_{D_2}, \dots, n_{D_p}$, the theorem still holds with $n_D$ replaced by $\max_i n_{D_i}$.
\end{remark}
\begin{proof}
Let $T = \{1, 2, \dots, n_D\}$ be the key set of the dimension table.
Define $t = n_D - 1$ and $k = 8\delta (n_D - 1)$. Also define $\alpha, \beta$, and the set $\mathcal{S}$ similarly as in \thmref{thm:2-table}. Let $d_1$ be the distribution on the fact table such that
\begin{itemize}
\item the first $(n_F - B)$ entries on the foreign key column takes values from $T\setminus\{n_D\}$ uniformly at random;
\item the remaining $B$ entries on the foreign key column take the value $n_D$.
\end{itemize}
Let $d_2$ be the distribution on the dimension table such that
\begin{itemize}
\item among the first $n_D - 1$ entries, given a uniformly chosen subset $S \in \mathcal{S}$ (of size $k$), those entries corresponding to $S$ are selected;
\item the last entry with key value $n_D$ is always selected.
\end{itemize}
The rest of the proof follows that in \thmref{thm:2-table}.
\end{proof}

\begin{theorem}[PK-FK join with selection \texttt{GROUP-BY}]\label{thm:pk-fk_group-by}
Let $\lambda \ge 1$. Following the setting in \thmref{thm:pk-fk} and notations in \thmref{thm:group-by}, then in order not to miss a set of groups whose total size is no larger than $n_F$ with $1-\delta$ probability, the length of the bit string that $\Phi$ assigns must be at least
\[
\frac{C'(\delta)}{2} \cdot n_D.
\]
\end{theorem}
\begin{remark}
For the star schema with one fact table of size $n_F$ and $p$ dimension tables of sizes $n_{D_1}, n_{D_2}, \dots, n_{D_p}$, the theorem still holds with $n_D$ replaced by $\max_i n_{D_i}$.
\end{remark}
\begin{proof}
The proof is a combination of those in \thmref{thm:group-by} (with $\lambda = 1$) and \thmref{thm:pk-fk}.
\end{proof}

\begin{theorem}[2-table \texttt{COUNT} with top-$K$ heavy hitters]\label{thm:heavy_hitter}
Let $\Phi$ be any scheme which assigns bit strings to database relations, so that there is a function $D$ such that given two relations $R_1$ and $R_2$ with the frequency vectors of the top-$K$ frequent elements (heavy hitters) $\mathbf{a} = (a_1, a_2, \dots, a_K)$ and $\mathbf{b} = (b_1, b_2, \dots, b_K)$ on the joining columns respectively, $D(\Phi(R_1), \Phi(R_2))$ gives an $(\epsilon, \delta)$-approximation for the \texttt{COUNT} result of $R_1 \Join R_2$ (join size), when an a priori lower bound $B$ is given on the join size. Denote the size of $R_1$ by $n_1$ (excluding the size of heavy hitters $\sum_{i=1}^K a_i$) and the size of $R_2$ by $n_2$ (excluding the size of heavy hitters $\sum_{i=1}^K b_i$) and assume $a_1 \ge a_2 \ge \dots \ge a_K \ge 1$ and $b_1 \ge b_2 \ge \dots \ge b_K \ge 1$. Then when the following conditions hold,
\begin{itemize}
\item
$a_K b_K \ge (\epsilon^2 + 2\epsilon) B$;
\item
$n_1 \ge \max\left(2, 1 + \frac{1}{\epsilon^2 + 2\epsilon}\right) a_K$;
\item
$n_2 \ge \max\left(2, 1 + \frac{1}{\epsilon^2 + 2\epsilon}\right) b_K$.
\end{itemize}
the length of the bit string that $\Phi$ assigns must be at least
\[
\frac{C(\delta)}{2} \cdot \max\left(m_1 \min\left(1, \frac{b_K}{(\epsilon^2 + 2\epsilon) B}\right), m_2 \min\left(1, \frac{a_K}{(\epsilon^2 + 2\epsilon) B}\right)\right),
\] where $m_1 = n_1 - \frac{B}{b_K}$, and $m_2 = n_2 - \frac{B}{a_K}$.
\end{theorem}
\begin{remark}
Here is a special case of the above theorem.
Let $\epsilon = 1$, $a_K = \frac{1}{4}n_1$, and $b_K = \frac{1}{4}n_2$. By assumption, we need $B \le \frac{1}{20}n_1 n_2$. We further assume $n_1 < n_2$. Then the theorem says that to get an $\left(1, \delta \right)$-approximation for the join size, the length of the bit string we need is at least $\frac{C(\delta)}{2} \min\left(n_2 - \frac{4B}{n_1}, \frac{4n_1n_2}{B} - 1\right)$.
\end{remark}
\begin{proof}
We only show the construction for $\frac{C(\delta)}{2} \cdot m_2 \min\left(1, \frac{a_K}{(\epsilon^2 + 2\epsilon) B}\right)$, the other part is symmetric.

Assume $\mathbf{a}$ and $\mathbf{b}$ have no intersections, i.e. the heavy hitters in the joining column of one relation do not appear in that of the other relation.
Define $k$ such that
\begin{itemize}
\item if $a_K < (\epsilon^2 + 2\epsilon) B$, let $k = \frac{a_K m_2}{(\epsilon^2 + 2\epsilon) B}$;
\item otherwise, let $k = m_2$.
\end{itemize}
In other words, $k = m_2 \min\left(1, \frac{a_K}{(\epsilon^2 + 2\epsilon) B}\right)$.
To make the theorem meaningful we need $k \ge 1$, and one can check this is always satisfied if $a_K b_K \ge (\epsilon^2 + 2\epsilon) B$ and $n_2 \ge \max\left(2, 1 + \frac{1}{\epsilon^2 + 2\epsilon}\right) b_K$ as we have assumed.
Define $t, \alpha, \beta$ and construct $T, \mathcal{S}$ similarly as in the proof of \thmref{thm:2-table} such that $T$ does not have intersection with the values in $\mathbf{a}$ and $\mathbf{b}$.

To construct the input distribution $d_1$ for table $R_1$, divide the entries in the joining column of $R_1$ into four parts and assign different distributions to them as follows:
\begin{enumerate}
\item the first $\sum_{i=1}^K a_i$ entries are fixed as the frequency vector $\mathbf{a}$ indicates;
\item the following $a_K$ entries take the same value uniformly randomly chosen from $T$;
\item if $B \ge a_K b_K$, then the following $\frac{B}{b_K}$ entries are divided into $\frac{B}{a_K b_K}$ blocks of size $a_K$, each has a distinct value different from values in $\mathbf{a}$, $\mathbf{b}$, and $T$; otherwise (i.e. $B < a_K b_K$), the following $a_K$ entries are filled with the same value different from values in $\mathbf{a}$, $\mathbf{b}$, and $T$;
\item the rest entries (if any) are assigned values not seen anywhere, without a single value appear more than $a_K$ times.
\end{enumerate}
For part (2)(3)(4) to hold we need $n_1 \ge a_K + \frac{B}{b_K}$ if $B \ge a_K b_K$ and $n_1 \ge 2a_K$ if $B < a_K b_K$, and one can check this is always satisfied if $a_K b_K \ge (\epsilon^2 + 2\epsilon) B$ and $n_1 \ge \max\left(2, 1 + \frac{1}{\epsilon^2 + 2\epsilon}\right) a_K$ as we have assumed.

To construct the input distribution $d_2$ for table $R_2$, divide the entries in the joining column of $R_2$ into three parts and assign different distributions to them as follows:
\begin{enumerate}
\item the first $\sum_{i=1}^K b_i$ entries are fixed as the frequency vector $\mathbf{b}$ indicates;
\item the following $m_2 = n_2 - \frac{B}{a_K}$ entries are partitioned into $k$ classes of the same size $\frac{m_2}{k}$ and entries in each class have the same value, and the $k$ different values are from a uniformly chosen set $S \in \mathcal{S}$;
\item if $B \ge a_K b_K$, then the last $\frac{B}{a_K}$ entries are divided into $\frac{B}{a_K b_K}$ blocks of size $b_K$, with values corresponding to the $\frac{B}{a_K b_K}$ blocks in part (3) of $R_1$; otherwise (i.e. $B < a_K b_K$), the last $\frac{B}{a_K}$ entries are filled with the same value corresponding to the single value in part (3) of $R_1$.
\end{enumerate}
For part (2), if $a_K < (\epsilon^2 + 2\epsilon) B$, $k = \frac{a_K m_2}{(\epsilon^2 + 2\epsilon) B}$ and thus the block size $\frac{m_2}{k} = \frac{(\epsilon^2 + 2\epsilon)B}{a_K}$ must be no more than $b_K$ (because otherwise the block size is larger than some of the top-$K$ heavy hitters), which is satisfied as we assume $a_K b_K \ge (\epsilon^2 + 2\epsilon) B$.

The rest of the proof follows that in \thmref{thm:2-table}.
\end{proof}

We know that with the full frequency information, the join size of two tables can be computed instantly. The following theorem demonstrates that when there are multiple tables, even with all the frequency vectors on the joining columns, computing the join size might be hard in the worst case.

\begin{theorem}[4-table chain join \texttt{COUNT} with all frequency vectors]\label{thm:frequency}
Let $\Phi$ be any scheme which assigns bit strings to database relations, so that there is a function $D$ such that given four relations $R_1, R_2, R_3, R_4$, each of size $n$, $D(\Phi(R_1), \Phi(R_2), \Phi(R_3), \Phi(R_4))$ gives an $(\epsilon, \delta)$-approximation for the \texttt{COUNT} result of chain join of the four tables (join size), when an a priori lower bound $B \le \frac{n^2}{\epsilon^2 + 2\epsilon}$ is given on the join size, then the length of the bit string must be at least
\[
\frac{C'(\delta)}{2} \cdot \left( n - \frac{\sqrt{B}}{n - \sqrt{(\epsilon^2 + 2\epsilon) B}} \right).
\]
\end{theorem}
\begin{remark}
The theorem can be easily generalized with more than four tables involved and different tables having different sizes.
\end{remark}
\begin{proof}

Suppose $R_1$ and $R_2$ join on column $C_{12}$, $R_2$ and $R_3$ join on column $C_{23}$, and $R_3$ and $R_4$ join on column $C_{34}$.
We use $R.C(v)$ to denote the number of entries of value $v$ on the column $C$ of relation $R$.
Let $k = 8\delta(n-y)$, $x = n - \sqrt{(\epsilon^2 + 2\epsilon) B}$ and $y = \frac{\sqrt{B}}{n - \sqrt{(\epsilon^2 + 2\epsilon) B}}$. Then define $t = n-y = \frac{k}{8\delta}, \alpha, \beta$, and construct $T, \mathcal{S}$ similarly as in \thmref{thm:2-table}.

\begin{table*}[t]
\begin{center}
\caption{An example of random relations from the input distribution where the join size is $B$.}
\label{tab:four-table}
\begin{tabular}{@{}c@{\hskip 0.2in}c@{\hskip 0.2in}c@{}}
\toprule
$\cdots$&$\cdots$& $R_1.C_{12}$\\
\midrule
&&$a$\\
&&$a$\\
&&$a$\\
&&$\vdots$\\
&&$a$\\
&&$0$\\
&&$\vdots$\\
&&$0$\\
\bottomrule
\end{tabular}
\qquad
\begin{tabular}{@{}c@{\hskip 0.2in}c@{\hskip 0.2in}c@{}}
\toprule
$R_2.C_{12}$&$\cdots$&$R_2.C_{23}$\\
\midrule
$a$&&$1$\\
$b$&&$2$\\
$b$&&$3$\\
\vdots&&$\vdots$\\
$b$&&$n-y$\\
$0$&&$0$\\
\vdots&&$\vdots$\\
$0$&&$0$\\
\bottomrule
\end{tabular}
\qquad
\begin{tabular}{@{}c@{\hskip 0.2in}c@{\hskip 0.2in}c@{}}
\toprule
$R_3.C_{23}$&$\cdots$&$R_3.C_{34}$\\
\midrule
$1$&&$c$\\
$2$&&$d$\\
$3$&&$c$\\
\vdots&&$\vdots$\\
$n-y$&&$c$\\
$0$&&$0$\\
\vdots&&$\vdots$\\
$0$&&$0$\\
\bottomrule
\end{tabular}
\qquad
\begin{tabular}{@{}c@{\hskip 0.2in}c@{\hskip 0.2in}c@{}}
\toprule
$R_4.C_{34}$&$\cdots$&$\cdots$\\
\midrule
$d$&&\\
$d$&&\\
$d$&&\\
\vdots&&\\
$d$&&\\
$0$&&\\
\vdots&&\\
$0$&&\\
\bottomrule
\end{tabular}
\end{center}
\end{table*}

First consider the following frequency vectors on the joining columns:
\begin{itemize}
\item
$R_1.C_{12}(a) = n - x,\ R_1.C_{12}(0) = x.$
\item
$R_2.C_{12}(a) = 1,\ R_2.C_{12}(b) = n-x-1,\ R_2.C_{12}(0) = y;\\ R_2.C_{23}(i) =1 (1 \le i \le n-y),\ R_2.C_{23}(0) = y.$
\item
$R_3.C_{23}(i) =1 (1 \le i \le n-y),\ R_3.C_{23}(0) = y;\\ R_3.C_{34}(c) = n-y-k,\ R_3.C_{34}(d) = k,\ R_3.C_{34}(0) = y.$
\item
$R_4.C_{34}(d) = n-x,\ R_4.C_{34}(0) = x.$
\end{itemize}

Then consider the following distributions on the relations which agree with the frequency vectors defined above:
\begin{description}
\item[$R_1.C_{12}$]
The first $n-x$ entries are filled with $a$'s. The remaining $x$ entries are filled with $0$'s.

\item[$R_2.C_{12}$]
The first $n-y$ entries are filled with $b$'s except for one $a$. The position of $a$ is uniformly randomly chosen among the $n-y$ possible positions. The remaining $y$ entries are filled with $0$'s.

\item[$R_2.C_{23}$]
The first $n-y$ entries are filled with $\{1, 2, \dots, n-y\}$. The remaining $y$ entries are filled with $0$'s.

\item[$R_3.C_{23}$]
The first $n-y$ entries are filled with $\{1, 2, \dots, n-y\}$. The remaining $y$ entries are filled with $0$'s.

\item[$R_3.C_{34}$]
The first $n-y$ entries are filled with $n-y-k$ of $c$'s and $k$ of $d$'s. The positions of $d$'s are chosen according to a uniformly randomly chosen subset $S \in \mathcal{S}$. The remaining $y$ entries are filled with $0$'s.

\item[$R_4.C_{34}$]
The first $n-x$ entries are filled with $d$'s. The remaining $x$ entries are filled with $0$'s.
\end{description}

See \tabref{tab:four-table} for an example of the random relations chosen from the input distributions constructed above. In this specific case, the join size is $B$
The rest of the proof is similar to that in \thmref{thm:2-table}.
\end{proof}

\bigskip

\section{Comparison between upper bounds and lower bounds}\label{sec:comparison}
In this section we compare the upper bounds and lower bounds on the number of samples needed to give an $(\epsilon, \delta)$-approximation of the cardinality (or the result of \texttt{COUNT} queries) of join operations among $p$ tables of size $n$, where $\epsilon$, $\delta$, and $p \ge 2$ are constant numbers.
In particular, we extended the result in \cite{Alon:1999:TJS:303976.303978} which studied two-table case.

We still assume an \emph{a priori} lower bound $B$ on the join size.
From the lower bounds in \thmref{thm:multi-table} we know that when $B = o(n^{p-1})$,
the amount of information any approximation scheme needs is linear in $n$, the size of the whole database.
Therefore, we assume $B = \Omega(n^{p-1})$.
The upper bound is shown in \propref{prop:bernoulli} by analyzing the simple \emph{Bernoulli sampling} using Chebyshev's inequality.
An overview of the results can be found in \tabref{tab:comparison}.
The upper bound and lower bound match each other (within a constant factor) when $B = \Theta(n^{k-1})$.

\begin{table*}[h!]\centering
\ra{1.3}
\caption{Amount of information needed to obtain a good estimate}\label{tab:comparison}
\begin{tabular}{@{}l@{\hskip 0.4in}c@{\hskip 0.4in}c@{\hskip 0.4in}c@{}}
\toprule
& 2-table~\cite{Alon:1999:TJS:303976.303978} & 3-table & p-table\\
\midrule
Bernoulli sampling &  $O\left(\frac{n^2}{B}\right)$ & $O\left(\frac{n^2}{\sqrt{B}}\right)$ & $O\left(\frac{n^2}{B^{1/(p-1)}}\right)$ \\
Lower bound &  $\Omega\left(\frac{\left(n - \sqrt{B}\right)^2}{B}\right)$ & $\Omega\left(\frac{\left(n - B^{1/3}\right)^3}{B}\right)$ & $\Omega\left(\frac{\left(n - B^{1/p}\right)^p}{B}\right)$ \\ \bottomrule
\end{tabular}
\end{table*}

Bernoulli sampling randomly selects each tuple in database relations with a uniform probability $q$, estimates the join size of $p$ relations by computing the join size of their samples and scaling the result by $q^{-p}$.

\begin{proposition}\label{prop:bernoulli}
Given $p$ relations $R_1, R_2, \dots, R_p$ of size $n$, suppose there is an a priori lower bound $B = \Omega(n^{p-1})$ on the join size of $\Join_i R_i$. The Bernoulli sampling scheme approximates the join size with constant relative error and constant error probability if the random sample has size $\frac{cn^2}{B^{1/(p-1)}}$, where $c$ is a constant determined by the desired accuracy and confidence.
\end{proposition}
\begin{proof}
We view the relations as a hypergraph $G = (V, E)$ where each tuple in a relation is a vertex.
Edges of this hypergraph correspond to the join result.
In other words, $|E| = |\Join_i R_i|$.
We note that $G$ is a $p$-uniform hypergraph, i.e., each edge contains exactly $p$ vertices.

We assume Bernoulli sampling selects each vertex with probability $q \ge \frac{1}{pn}$.
For each edge $e$, denote by $X_e$ to be the indicator variable such that $X_e = 1$ if all the vertices of $e$ are selected by Bernoulli sampling.
Define $X = \sum_{e \in E}X_e$.
Then $X$ is a random variable and its expectation
\[
\expct{X} = \sum_{e \in E} \expct{X_e} = |E| q^{p}
\]
since each edge is selected with probability $q^{p}$.
To compute the the variance $\var{X}$ we use the formula
\[
\var{\sum_{e \in E} X_e} = \sum_{e \in E} \var{X_e} + \sum_{e \neq e'} \cov{X_e, X_{e'}}.
\]
In a hypergraph, the degree of a vertex $v$ is the number of edges containing $v$.
Denote the degree of the $i$th vertex by $d_i$.
The covariance between any $X_e$ and $X_{e'}$ with $e \neq e'$ is at most $q^{p+1}$, and it is nonzero if and only if $e$ and $e'$ share some vertices.
Therefore, we have
\[
\var{\sum_{e \in E} X_e} \le |E| \left( q^p - q^{2p} \right) + \sum_{i=1}^{pn} d_i(d_i - 1) q^{p+1}.
\]
It is not hard to see that in the $p$-uniform hypergraph $G$,  $\sum_{i=1}^{pn} d_i = p |E|$.
Since $d_i \le n^{p-1}$, we have $\sum_{i=1}^{pn} d_i(d_i - 1) \le n^{p-1} \sum_{i=1}^{pn} d_i = pn^{p-1} |E|$. Recall that $q \ge \frac{1}{pn}$. Therefore,
\[
\var{\sum_{e \in E} X_e} \le p|E| nq^{p+1} + p |E| n^{p-1} q^{p+1} \le 2p|E|n^{p-1}q^{p+1}.
\]

Note that when $\expct{X}^2 \ge \alpha \var{X}$ for a constant $\alpha$, we can apply the Chebyshev's inequality to obtain a constant relative error with constant error probability.
This happens when
\[
q^{p-1} \ge \frac{2p n^{p-1}}{|E|},
\]
or in other words, the number of samples
\[
pnq = \Omega\left(\frac{n^2}{|E|^{1/(p-1)}}\right) = \Omega\left(\frac{n^2}{B^{1/(p-1)}}\right).
\]
\end{proof}

\section*{Acknowledgement}
This work benefitted from discussions with Surajit Chaudhuri and Vivek Narasayya.

% \bigskip

\bibliography{reference}{}
\bibliographystyle{alpha}

\end{document}